\documentclass[graybox]{svmult}
\usepackage[utf8]{inputenc}  %encodage d'entrée du fichier

\usepackage{mathptmx}       % selects Times Roman as basic font
\usepackage{helvet}         % selects Helvetica as sans-serif font
\usepackage{courier}        % selects Courier as typewriter font
\usepackage{type1cm}        % activate if the above 3 fonts are
\usepackage{makeidx}         % allows index generation
\usepackage{graphicx}        % standard LaTeX graphics tool
\usepackage{multicol}        % used for the two-column index
\usepackage[bottom]{footmisc}% places footnotes at page bottom
\usepackage{mathtools}  % ce package charge amsmath et corrige des bugs
\usepackage{amssymb} %  utile pour plein de polices de maths (notamment mathcal)
\makeindex             % used for the subject index
\begin{document}

\title*{An anisotropic model for global climate data}
\subtitle{\emph{Un modello anisotropico per i dati climatici globali}}
\author{Nil Venet\protect\footnote{nil.venet@unibg.it, https://nilvenet.org/} and Alessandro Fassò\protect\footnote{alessandro.fasso@unibg.it}}
%\institute{Nil Venet \at Bergamo University, viale Marconi 5 - 24044 Dalmine (BG) - Italia , \email{nil.venet@unibg.it}
%\and Alessandro Fassò \at Bergamo University, viale Marconi 5 - 24044 Dalmine (BG) - Italia  \email{alessandro.fasso@unibg.it}}
\maketitle

\vspace{-2cm}
\emph{The following article has been accepted to the \emph{Statistical models for climate data} invited specialized session at the 2019 conference of the Italian Statistics Society  \emph{"SIS2019"}.}

\vspace{2cm}

\abstract{We present a new, elementary way to obtain  axially symmetric Gaussian processes on the sphere, in order to accommodate for the directional anisotropy of global climate data in geostatistical analysis.}
\abstract{\emph{Presentiamo un nuovo modo elementare per ottenere processi gaussiani assialmente simmetrici sulla sfera, al fine di accogliere l'anisotropia direzionale dei dati climatici nelle analisi geostatistiche su scala globale.}}

\section{Introduction}
\label{sec:intro}
Following the increasing amount of measurements and computing power available, geostatistics has evolved in the last decades to consider massive and global scale datasets. This adaptation presents new challenges, such as taking into account the spherical nature of Earth, and the considerable non-stationarity of global datasets, while maintaining computing feasibility.

We focus in this article on Gaussian process based geostatistics for climate data. Isotropic models on the sphere, which assume that the data has the same statistical behavior in every direction, are well understood and the literature gives rich classes of covariances (see for example \cite{Gneiting_2013} and the reviews \cite{Jeong_al_2017} and \cite{Porcu_al_2018}).

However, climate data exhibits various nonstationary effects at the global scale. In particular, it is clear that the variability of climate variables is higher in the north-south direction, than in the east-west. Jones \cite{Jones_1963} notices this issue and defines axially symmetric Gaussian processes on the sphere, which are stationary in longitude only. He also characterizes their decomposition on the spherical harmonic basis, which Stein \cite{Stein_2007} truncates to a finite order to carry on a statistical analysis of a Total Ozone dataset at a global scale. While very flexible, Stein's model involves the estimation of a large number of coefficients, and still doesn't fit the local behaviour of the Ozone dataset. This calls for explicit axially symmetric covariances that relies on a small number of parameters. Such covariance functions are obtained through partial differentiation of isotropic processes by Jun and Stein \cite{Jun_Stein_2007,Jun_Stein_2008}, and further studied by
Hitczenko and Stein \cite{Hitczenko_Stein_2012}. Recently Porcu et al. \cite{Porcu_al_2019} proposed to modify variograms of isotropic covariances to obtain axially symmetric analogues, while Huang et al. \cite{Huang_al_2012} consider products of separated covariances on latitudes and longitudes.

Our approach consists in considering the product of an isotropic covariance with an additional covariance on latitudes. Compared to Huang et al.,\cite{Huang_al_2012}, our method has the advantage of giving Gaussian processes that are continuous over the whole sphere, including the poles.

In Section \ref{sec:data}, we give a description of two climatic datasets, which motivate our work and allow performance evaluation of models. Section \ref{sec:modeling} recalls some generalities on Gaussian processes on the sphere and introduce our model. In Section \ref{sec:persp}, we describe our ongoing work to compare model performances and properties to the existing literature.

\section{Climatic datasets}
\label{sec:data}
In this section, we describe two important sources of climatic data, which motivates our study.

The radiosonde observations database from NOAA\footnote{US National Oceanic and Atmospheric Administration} ESRL\footnote{Earth System Research Laboratory} gathers worldwide radiosonde observations, that are considered as an important atmospheric observation standard. We propose to use it to assess model performances on real data.

We also consider the reanalysis ERA-Interim archive from ECMWF\footnote{European Centre for Medium-Range Weather Forecasts}, that does not present the drawbacks of real data (missing data, inhomogeneous spatial covering) but rather constitutes a resource whose coherence and completeness allows to extract tailor-fitted datasets to assess a virtually unlimited variety of hypothesis.

\subsection{Radiosonde observations}
As mentioned above, the NOAA ESRL database gathers worldwide radiosonde observations.
There are around 900 ground stations performing radiosonde observations in the world. Ascending balloons equipped with measurement devices are launched, radio-transmitting the observations of meteorological variables about every few seconds to the ground.

These measurements are stored at a restricted number of altitudes levels, corresponding to pressure levels: there are 22 \emph{mandatory levels} (from the ground up to 1 hPa) common to every sounding, to which are added a varying number of \emph{significant levels} (28 on the average), that depend on the launch and correspond to important variations in measured temperature or dew point depression. Missing data at the highest mandatory levels are frequent as the balloon explodes between 30 and 35km of altitude.

There are typically two launches per day in a station, occurring at times close to 00 and 12 UTC, but additional and missing launches occur. Data is available publicly and can be downloaded from the NOAA ESRL servers.

\subsection{The ERA-Interim reanalysis}
Produced by the ECMWF, the ERA-Interim archive is a reanalysis of global atmospheric data from 1979, updated in real-time to nowadays.

This model output is built from the input of numerous meteorological datasets of various nature (mostly satellites measurements, but also radiosondes, land, boat, planes,... measurements), which are assimilated and extrapolated following a 4-dimensional variational analysis.

Era-Interim provides values for a considerable number of meteorological variables. Atmospheric variables are given at a vertical resolution of sixty pressure levels (from the ground to 0.1 hPa) on the horizontal reduced Gaussian grid N128, which is regular in latitude but has a variable longitudinal resolution to maintain a ground resolution of approximately 79km (see details in \cite{ERA_grid}). Time resolution is 6h. Surface variables are given with the same horizontal display and a time resolution of 3 hours. In both cases it is also possible to download data which have been extrapolated on regular latitude/longitude grids at various resolutions.

Advantages of this archive consist in its completeness, its coherence in the methodology and in the physical consistency of the data. From this, it is possible to tailor-fit datasets to a virtually unlimited number of applications. 

Specifications of the archive are given in \cite{ERAreport}, while the reanalysis process is described in \cite{ERAarticle}. Data is available publicly and can be downloaded from the ECMWF servers.

\section{Spatial Modeling}
\label{sec:modeling}
We briefly recall generalities about Gaussian process based geostatistics (for which we refer to Cressie \cite{Cressie_book_1992} for more details) and Gaussian processes on the sphere before to introduce our new model.
\subsection{Gaussian process geostatistics}

The classical approach in geostatistics, given observed data $y_1,\cdots,y_n \in \mathbb{R}$ at locations $x_1,\cdots,x_n$ in some space $S$ consists in considering a Gaussian process $(Y_x)_{x\in S}$ indexed by $S$, and estimating the value of the variable of interest at a new location $x^*$ by the Kriging estimator
\begin{equation}\hat{Y}(x^*) = \mathbb{E}(Y(x^*)|Y(x_1) = y_1, \cdots, Y(x_n) = y_n ).\end{equation}
The quality of this prediction critically relies on the choice of the Gaussian process $Y$, whose statistical behavior should fit as much as possible the data of interest. The classical approach consists in selecting a Gaussian process from a parametric class $Y_\theta$, where $\theta$ is a vector of parameters, which is typically done by maximizing the likelihood. The key asset of this approach is that prediction comes with a statistical model of the data, which allows in particular the computation of the conditional variance
\begin{equation}Var(Y(x^*)|Y(x_1) = y_1, \cdots, Y(x_n) = y_n ),\end{equation}
that quantifies the prediction uncertainty at the location $x^*$.

Due to the $O(n^3)$ numerical complexity of both Kriging prediction and likelihood evaluation, this approach proves itself computationally heavy -- if not unfeasible -- on massive datasets. There exists an important literature on adaptations to large datasets. See for example \cite{Heaton_al_2018} for a comparison of methods a single case study.

However in the majority of methods, the primary ingredient remain Gaussian processes, and finding classes of Gaussian processes that suit practical applications constitute an entire field of research.

Let us recall that the statistical properties of a Gaussian process are entirely characterized by its expectation and covariance functions. As the choice of a a mean function is free from any constrains, in this article all Gaussian fields will be assumed centered. In this setting, considering a Gaussian process is equivalent to considering a covariance function.

In contrast with the mean function, not every function is admissible as a covariance. Let us recall that a function $F : X\times X \rightarrow \mathbb{R}$ is the covariance of a Gaussian process $(X_x)_{x\in X}$ if and only if it is symmetric and \emph{positive definite}, that is
\begin{equation}\forall \lambda_1,\cdots,\lambda_n \in \mathbb{R}, \forall x_1 \cdots x_n \in X, ~ \sum_{i,j=1}^n \lambda_i \lambda_j F(x_i,x_j) \geq 0 .\end{equation}
In this case we will often say that $F$ is a \emph{valid covariance}.

The finding of classes of covariance functions that give Gaussian properties which fit data from application is the first, necessary step of Gaussian process modeling.
 
\subsection{Gaussian processes on the sphere}
\label{subsec:sphere_gen}
In order to model variables on the Earth, we consider Gaussian processes indexed by the sphere $\mathbb{S}^2$, which we will take of radius $1$. To a given point $x$ on the sphere are associated its longitude $\theta_x \in [-\pi,\pi]$ and latitude $\varphi_x \in [-\pi/2,\pi/2]$, in radians. Observe that $\theta_x$ and $\varphi_x$ are well defined for every $x$ distinct from the north or south poles, for which $\varphi_x$ is $\pi/2$ (resp. $-\pi/2$) but $\theta_x$ can take any value. This singularity of the spherical coordinates has some practical consequences when it comes to define a covariance on the whole sphere, as we will see in Section \ref{subsec:model}.

Given a Gaussian random field $X_{x\in \mathbb{S}^2}$ and $x,y \in \mathbb{S}^2$, we will write $K_X(x,y)$ or $K_X(\theta_x,\theta_y,\varphi_x,\varphi_y)$ for its covariance function.

As in the Euclidean case, a simplifying assumption is to consider stationary models.A A Gaussian process on the sphere is called \emph{isotropic} if $K_X(x,y) = F(d(x,y))$, where $d(x,y)$ is the great circle distance on the sphere, given by the expression:
\begin{equation}d(x,y) = \cos^{-1}{( \sin{\varphi_1} \cdot \sin{\varphi_2} + \cos{\varphi_1} \cdot \cos{\varphi_2} \cdot \cos{(\theta_2 - \theta_1)} )}.\end{equation}

We refer to Gneiting\cite{Gneiting_2013}, Porcu et al. \cite{Porcu_al_2018}, Jeong et al. \cite{Jeong_al_2017} and references within for a state of the art on isotropic Gaussian processes.

However at the global scale, climatic data exhibits important nonstationarity effects: in particular, it is clear that climatic variables are typically correlated at a shorter range in the latitudinal direction than in the longitudinal. Jones \cite{Jones_1963} proposes to address this issue and defines \emph{axially symmetric} Gaussian processes, which are stationary only in longitude variable, that is to say their covariance verifies
\begin{equation}K_X(x,y) = F(\theta_x-\theta_y \mod 2\pi,\varphi_x,\varphi_y).\end{equation} Furthermore an axially symmetric Gaussian process is said to be \emph{latitudinally reversible} if
\begin{equation}F(\theta_x-\theta_y \mod 2\pi,\varphi_x,\varphi_y) = F(\theta_x-\theta_y \mod 2\pi,\varphi_y,\varphi_x).\end{equation}

\subsection{A new class of axially symmetric covariances}
\label{subsec:model}
In the Euclidean case, a natural idea to obtain anisotropic covariances is to consider the product of covariances with respect to different coordinates. For example considering for $x,y \in \mathbb{R}^2,$
\begin{equation}K(x,y) = \sigma e^{-\frac{|x_1-y_1|}{r_1}}\cdot \ e^{-\frac{|x_2-y_2|}{r_2}}\end{equation}
with distinct scale parameters $r_1$ and $r_2$ yields a model that is correlated at a longer range in a chosen direction. On the sphere one might consider in a similar way products of covariances in latitude and longitude, such as for $x,y\in \mathbb{S}^2,$
\begin{equation}\label{eq:exp_sep}K(x,y) = \sigma e^{-\frac{|\varphi_x-\varphi_y|}{r_{\varphi}}}\cdot \ e^{-\frac{|\theta_x-\theta_y|}{r_{\theta}}},\end{equation}
but in this case, since the longitude coordinate can take any value at the poles of the sphere (see Section \ref{subsec:sphere_gen}), the covariance is not well defined when $x$ or $y$ is one of the two poles. Moreover, there is no way to chose a value for $K$ at these points to obtain a continuous kernel, since when $x$ (or $y$) closes to one of the pole, the value taken by $K(x,y)$ depend on the way $x$ approaches the pole.

Discontinuity of the covariance kernel yields Gaussian fields that are discontinuous in $L^2$ and whose realizations are not almost surely continuous. This is unwanted in most applications where variables exhibits a continuous behavior over the sphere, such as climatic applications. In the case of \eqref{eq:exp_sep} we obtain an expected singular behavior at the poles.

To address this issue, we consider instead of \eqref{eq:exp_sep} a covariance 
\begin{equation}\label{eq:exp_our}K(x,y) = \sigma e^{-\frac{d(x,y)}{r_{iso}}}\cdot \ e^{-\frac{|\varphi_x-\varphi_y|}{r_{\varphi}}}.\end{equation}
Notice that since $(x,y) \mapsto d(x,y)$ and $x \mapsto \varphi_x$ are continuous function on the sphere, this time $K$ is continuous everywhere. The term $e^{-\frac{|\varphi_x-\varphi_y|}{r_{\varphi}}}$ can be seen as an additional decorrelation to the isotropic covariance $\sigma e^{-\frac{d(x,y)}{r_{iso}}}$ in the latitudinal direction.

Extending this idea to kernels other than the exponential, we obtain the following class of covariances.

\begin{theorem} Let $K_{iso}$ be an isotropic covariance on the sphere and $K_\varphi$ be a covariance on $[-\pi,\pi]$.
	
	The kernel defined by
	\begin{equation} \label{eq:general_form} K(x,y) = K_{iso}(x,y) \cdot K_\varphi(\varphi_x,\varphi_y)\end{equation}
	is a latitudinally reversible, axially symmetric covariance on the sphere.
	
	Furthermore, if $K_{iso}$ and $K_\varphi$ are continuous, $K$ is continuous on the whole sphere, and as such, a Gaussian field with covariance $K$ is continuous in $L^2$ sense, and has almost surely continuous trajectories.
\end{theorem}
\begin{proof} As product of two valid covariances, $K$ is a valid covariance (see Schur product theorem in Zhang \cite{Zhang_book_2005}). Axial symmetry of $K$ is a direct consequence of the isotropy of $K_{iso}$ and the fact that $K_\varphi$ does not depend on $\theta_x$ and $\theta_y$.
\end{proof}

In contrast with Jun and Stein \cite{Jun_Stein_2007,Jun_Stein_2008}, whose approach seems to give richer classes of covariances, allowing for complex interactions between variables, our method permits to easily modify the classical isotropic covariances to obtain axially symmetric analogues. Porcu et al. \cite{Porcu_al_2019} attain a similar goal, but focusing on the variograms.
In comparison with Huang et al. \cite{Huang_al_2012}, our covariances present the advantage to be continuous on the whole sphere.

\section{Perspectives: performances and properties of axially symmetric models}
\label{sec:persp}
In an ongoing work, we assess the performances of our new class of axially symmetric covariances, compared to the propositions of Jun and Stein \cite{Jun_Stein_2007,Jun_Stein_2008}, Huang \cite{Huang_al_2012}, Porcu et al. \cite{Porcu_al_2019}. We take advantage of the two data sources that were presented in Section \ref{sec:data} to assess their performances in prediction, training on real data and assessing prediction with a global covering.

We expect that the simple expression \eqref{eq:general_form} of our covariances allows us to study their differentiability at the origin in every direction, allowing for Gaussian processes with homogeneous trajectory regularity, or at the contrary for models which exhibit a regularity that depends on the direction of the trajectory.
\clearpage
\bibliographystyle{spmpsci}
\bibliography{biblio.bib}
\end{document}